\documentclass{sig-alternate}[9pt]

\usepackage{caption}
\usepackage{subcaption}
\usepackage{url}
\DeclareMathAlphabet{\mathpzc}{OT1}{pzc}{m}{it}

\newtheorem{theorem}{Theorem}
\newtheorem{lemma}[theorem]{Lemma}

\newtheorem{corollary}[theorem]{Corollary}

\newcommand{\comment}[1]{}
\renewcommand{\Pr}{\mathop{\bf Pr}\nolimits}
\newcommand{\E}{\mathbb{E}}

\begin{document}

\title{Balanced Allocations and Double Hashing \titlenote{}}
\numberofauthors{1} %
\author{
\alignauthor
Michael Mitzenmacher\titlenote{Supported in part by NSF grants CCF-0915922, IIS-0964473, and CNS-1011840.}\\
       \affaddr{Harvard University}\\
       \affaddr{School of Engineering and Applied Sciences}\\
       \email{michaelm@eecs.harvard.edu}
}

\maketitle

\begin{abstract}
With double hashing, for an item $x$, one generates two hash values
$f(x)$ and $g(x)$, and then uses combinations $(f(x) +i g(x)) \bmod n$
for $i=0,1,2,\ldots$ to generate multiple hash values from the initial
two.  We show that the performance difference between double hashing
and fully random hashing appears negligible in the standard balanced
allocation paradigm, where each item is placed in the least loaded of
$d$ choices, as well as several related variants.  We perform an
empirical study, and consider multiple theoretical approaches.  While
several techniques can be used to show asymptotic results for the
maximum load, we demonstrate how fluid limit methods explain 
why the behavior of double hashing and fully random
hashing are essentially indistinguishable in this context.
\end{abstract}

\section{Introduction}
\label{sec:introduction}
The standard balanced allocation paradigm works as follows: suppose
$n$ balls are sequentially placed into $n$ bins, where each ball is
placed in the least loaded of $d$ uniform independent choices of the
bins.  Then the maximum load (that is, the maximum number of balls in
a bin) is $\frac{\log \log n}{\log d} + O(1)$,
much lower than the $\frac{\log n}{\log \log n} (1 +o(1))$ obtained
where each ball is placed according to a single uniform choice \cite{ABKU}.

The assumption that each ball obtains $d$ independent uniform choices
is a strong one, and a reasonable question, tackled by several other
works, is how much randomness is needed for these types of results
(see related work below).  Here we consider a novel approach,
examining balanced allocations in conjunction with {\em double hashing}.  In
the well-known technique of standard double hashing for open-addressed hash tables,
the $j$th ball obtains two hash values, $f(j)$ and $g(j)$.
For a hash table of size $n$, $f(j) \in [0,n-1]$ and $g(j) \in
[1,n-1]$.  Successive locations $h(j,k) = f(j)+kg(j) \bmod n$,
$k=0,1,2,\ldots...$, are tried until an empty slot is found.
As discussed later in this introduction, double hashing is extremely conducive
to both hardware and software implementations and is used in many deployed systems.

In our context, we use the double hashing approach somewhat
differently.  The $j$th ball again obtains two hash values $f(j)$ and
$g(j)$.  The $d$ choices for the $j$th ball are then given by $h(j,k)
= f(j) + k g(j) \bmod n$, $k=0,1,\ldots,d-1$, and the ball is placed
in the least loaded.  We generally assume that $f(j)$ is uniform over
$[0,n-1]$, $g(j)$ is uniform over all numbers in $[1,n-1]$ relatively
prime to $n$, and all hash values are independent.  (It is convenient
to consider $n$ a prime, or take $n$ to be a power of 2 so that the
$g(j)$ are uniformly chosen random odd numbers, to ensure the $h(j,k)$
values are distinct.)

It might appear that limiting the space of random choices available to
the balls in this way might change the behavior of this random process
significantly.  We show that this is not the case both in theory
and in practice.  Specifically, by
``essentially indistinguishable'', we mean that, empirically, for any
constant $i$ and sufficiently large $n$ the fraction of bins of load $i$ is well
within the difference expected by experimental variance for the two methods.
Essentially indistinguishable means that in practice for
even reasonable $n$ one cannot readily distinguish the two methods.
By ``vanishing'' we mean that, analytically, 
for any constant $i$ the asymptotic fraction of bins of load $i$ 
for double hashing differs only by $o(1)$ terms from fully independent choices
with high probability.  A related key result is that $O(\log \log n)$ bounds on the maximum load hold
for double hashing as well.
Surprisingly, the difference between $d$ fully independent choices and $d$ choices using
double hashing are essentially indistinguishable for sufficiently
large $n$ and vanishing asymptotically.  
\footnote{To be clear, we do not mean that there is {\em no} difference
between double hashing and fully random hashing in this setting;  there clearly
is and we note a simple example further in the paper.  As we show, analytically
in the limit for large $n$ the difference is vanishing (Theorem~\ref{mainthm} and Corollary~\ref{cormain}), and for finite
$n$ the results from our experiments demonstrate the difference is essentially
indistinguishable (Section~\ref{sec:sims}).}

As an initial example of empirical results, Table~\ref{table_example} below shows the
fraction of bins of load $x$ for various $x$ taken over 10000 trials,
with $n=2^{14}$ balls thrown into $n$ bins using $d=3$ and $d=4$ choices, using
both double hashing and fully random hash values (where for our proxy for
``random'' we utilize the standard approach of simply generating successive random values using the drand48
function in C initially seeded by time).  Most values are given to five decimal
places.  The performance difference is essentially indistinguishable, well within
what one would expect simply from variance from the sampling process.

\begin{table*}[thp]
\begin{subfigure}[h]{0.45\textwidth}
\centering
\begin{tabular}{|c|c|c|} \hline
Load & Fully Random & Double Hashing \\ \hline
0 & 0.17693 & 0.17691 \\
1 & 0.64664 & 0.64670 \\
2 & 0.17592 & 0.17589 \\
3 & 0.00051 & 0.00051 \\ \hline
\end{tabular}
\caption{3 choices, $n=2^{14}$ balls and bins}
\end{subfigure}
\qquad
\begin{subfigure}[h]{0.45\textwidth}
\centering
\begin{tabular}{|c|c|c|} \hline
Load & Fully Random & Double Hashing \\ \hline
0 & 0.14081 & 0.14081 \\
1 & 0.71840 & 0.71841 \\
2 & 0.14077 & 0.14076 \\
3 & $2.25 \cdot 10^{-5}$ & $2.29\cdot 10^{-5}$ \\ \hline
\end{tabular}
\caption{4 choices, $n=2^{14}$ balls and bins}
\end{subfigure}
\caption{An initial example showing the performance of double hashing compared to fully random hashing.
In our tables, the row with load $x$ gives the fraction of the bins that
have load $x$ over all trials.  So over 10000 trials of throwing   
$n=2^{14}$ balls into $2^{14}$ bins using 3 choices and double hashing, the fraction of bins with load 0
was $0.17691$. \label{table_example}}
\end{table*}

More extensive empirical results appear in Appendix~\ref{sec:sims}.  In
particular, we also consider two extensions to the standard paradigm:
V\"{o}cking's extension (sometimes called $d$-left hashing), where the
$n$ bins are split into $d$ subtables of size $n/d$ laid out left to
right, the $d$ choices consist of one uniform independent choice in
each subtable, and ties for the least loaded bin are broken to the
left \cite{vocking}; and the continuous variation, where the bins
represent queues, and the balls represent customers that arrive as a
Poisson process and have exponentially distributed service
requirements \cite{Mitzenmacher}. We again find empirically that
replacing fully random choices with double hashing leads to essentially
indistinguishable results in practice.\footnote{We encourage the reader to 
examine these experimental results.  However, because we recognize
some readers are as a rule uninterested in experimental results, we have moved
them to an appendix.}

In this paper, we provide theoretical results explaining why this would be the
case.  There are multiple methods available that can yield $O(\log
\log n)$ bounds on the maximum load when $n$ balls are thrown into $n$
bins in the setting of fully random choices.  We therefore first
demonstrate how some previously used methods, including the layered
induction approach of \cite{ABKU} and the witness tree approach of
\cite{vocking}, readily yield $O(\log \log n)$ bounds; this asymptotic
behavior is, arguably, unsurprising (at least in hindsight).  We then
examine the key question of why the difference in empirical
results is vanishing, a much stronger requirement.  For the case of fully
random choices, the asymptotic fraction of bins of each possible load
can be determined using fluid limit methods that yield a family of
differential equations describing the process behavior
\cite{Mitzenmacher}.  It is not a priori clear, however, why the
method of differential equations should necessarily apply when using
double hashing, and the primary result of this paper is to explain why
it in fact applies.  The argument depends technically on the idea that
the ``history'' engendered by double hashing in place of $d$ fully
random hash functions has only a vanishing (that is, $o(1)$) effect
on the differential equations that correspond to the limiting behavior of the bin
loads.  We believe this resolution suggests that double
hashing will be found to obtain the same results as fully random hashing
in other additional hash-based structures, which may be important in practical
settings.

We argue these results are important for multiple reasons.  First, we
believe the fact that moving from fully random hashing to double
hashing does not change performance for these particular balls and
bins problems is interesting in its own right.  But it also has
practical applications; multiple-choice hashing is used in several
hardware systems (such as routers), and double hashing both requires
less (pseudo-)randomness and is extremely conducive to implementation
in hardware \cite{chunkstash,heileman2005caching}.  (As we discuss below, it may also be useful in software
systems.)  Both the fact that double hashing does not change
performance, and the fact that one can very precisely determine the
performance of double hashing for load balancing simply using the same
fluid limit equations as have been used under the assumption of fully
random hashing, are therefore of major importance for designing
systems that use multiple-choice methods (and convincing system
designers to use them).  Finally, as mentioned, these results suggest
that using double hashing in place of fully random choices may
similarly yield the same performance in other settings that make use
of multiple hash functions, such as for
cuckoo hashing or in error-correcting codes, offering the same
potential benefits for these problems. We have explored this issue
further in a subsequent (albeit already published) paper \cite{MT},
where there remain further open questions.  In particular, we have not
yet found how to use the fluid limit analysis used here for these other
problems.  

Finally, it has been remarked to us that all of our arguments here apply beyond
double hashing; any hashing scheme where the $d$ choices for a ball
are made so that they are pairwise independent and uniform would yield
the same result by the same argument.  That is, if for a given ball with $d$ choices 
$h_1,h_2,\ldots,h_d$, for any distinct bins $b_1$ and $b_2$ we have for all $1\leq i,j \leq d, i \neq j$:
$$\Pr(h_i = b_1) = 1/n \mbox{ and }$$
$$\Pr(h_i = b_1 \mbox{ and } h_j = b_2) = \frac{1}{{n \choose 2}},$$
then our results apply.  Unfortunately, we do not know of any actual
scheme besides double hashing in practical use with these properties;
hence we focus on double hashing throughout.  

\subsection{Related Work}

The balanced allocations paradigm, or the power of two choices, has
been the subject of a great deal of work, both in the discrete balls
and bins setting and in the queueing theoretic setting.  See, for
example, the survey articles \cite{KSurvey,TwoSurvey} for references
and applications. 

Several recent works have considered hashing variations that utilize
less randomness in place of assuming perfectly random hash functions;
indeed, there is a long history of work on universal hash functions
\cite{CarterWe79}, and more recently min-wise independent hashing
\cite{mwi}.  Specific recent related works include results on standard
one-choice balls and bins problems \cite{SHF}, hashing with linear
probing with limited independence \cite{ppr}, and tabulation hashing
\cite{pt}; other works involving balls and bins with less randomness
include \cite{godfrey,peres}.  
As another example, Woelfel shows that a variation of V\"{o}cking's results
hold using simple hash functions that utilize a collection of $k$-wise
independent hash functions for small $k$, and a random vector
requiring $o(n)$ space \cite{Woelfel}.  

Another related work in the balls and bins setting
is the paper of Kenthapadi and Panigrahy \cite{kpbalanced}, who
consider a setting where balls are not allowed to choose any two bins,
but are forced to choose two bins corresponding to an edge on an
underlying random graph.  In the same paper, they also show that two
random choices that yield $d$ bins are sufficient for similar $O(\log
\log n)$ bounds on maximum loads that one obtains with $d$ fully
random choices, where in their case each random choice gives a
contiguous block of $d/2$ bins.

Interestingly, the classical question regarding the average length of an unsuccessful 
search sequence for standard double hashing in an open address hash table when the table load is a constant $\alpha$ has been shown to be, up to
lower order terms, $1/(1-\alpha)$, showing that double hashing has
essentially the same performance as random probing (where each ball would have its own random permutation of
the bins to examine, in order, until finding an empty bin)
when using traditional hash tables \cite{bradford2007probabilistic,guibas1978analysis,lueker1993more}.
These results appear to have been derived using different techniques than we utilize here;
it could be worthwhile to construct a general analysis that applies for both schemes.

A few papers have recently suggested using double hashing in schemes
where one would use multiple hash functions and shown
little or no loss in performance.  For Bloom filters, Kirsch and
Mitzenmacher \cite{kirsch:bbb}, starting from the empirical analysis by
Dillinger and Manolios \cite{dillinger3312bfp}, prove that using double hashing has
negligible effects on Bloom filter performance.  
This result is closest in spirit
to our current work;  indeed, the type of analysis here can be used to provide an alternative argument
for this phenomenon, although the case of Bloom filters is inherently simpler.
Several available online implementations of Bloom filters now use this approach, 
suggesting that the double hashing approach can be significantly beneficial in software as well as 
hardware implementations.\footnote{See, for example, \url{http://leveldb.googlecode.com/svn/trunk/util/bloom.cc},
\url{https://github.com/armon/bloomd}, and \url{http://hackage.haskell.org/packages/archive/bloomfilter/1.0/doc/html/bloomfilter.txt}.}
Bachrach and Porat
use double hashing in a variant of min-wise independent sketches \cite{bachrach2010fast}.
The reduction in randomness stemming from using double hashing to generate
multiple hash values can be useful in other contexts.  For
example, it is used in \cite{MVad} to improve results where pairwise
independent hash functions are sufficient for suitably random data;
using double hashing requires fewer hash values to be generated (two in place
of a larger number), which means less randomness in the data is required. 
Finally, in work subsequent to the original draft of this paper \cite{MT}, we have empirically examined
double hashing for other algorithms such as cuckoo hashing, and again found essentially no empirical 
difference between fully random hashing and double hashing in this and other contexts.  However,
theoretical results for these settings that prove this lack of difference are as of yet very limited.


Arguably, the main difference between our work and other related work
is that in our setting with double hashing we find the empirical
results are essentially indistinguishable in practice, and we focus 
on examining this phenomenon.  

\section{Initial Theoretical Results}

We now consider formal arguments for the excellent behavior for double
hashing.  We begin with some simpler but coarser arguments that have
been previously used in multiple-choice hashing settings, based on
majorization and witness trees.  While our witness tree argument
dominates our majorization argument, we present both, as they may be
useful in considering future variations, and they highlight how these
techniques apply in these settings.  
In the following section, we then
consider the fluid limit methodology, which best captures the result
we desire here, namely that the load distributions are essentially the
same with fully random hashing and double hashing.  However, the fluid
limit methodology captures results about the fraction of bins with
load $i$, for every constant value $i$, and does not readily provide $O(\log \log n)$ 
bounds (without specialized additional work, which often depend on the techniques
used below).  The reader conversant with balanced allocation results 
utilizing majorization and witness trees may choose to skip
this section.    

\subsection{A Majorization Argument}
We first note that using double hashing with two choices and using random
hashing with two distinct hash values per ball are equivalent.  With
this we can provide a simple argument, showing the seemingly obvious fact
that using double hashing with $d > 2$ choices is at least as good as
using 2 random choices.  This in turn shows that 
double hashing maintains $\log \log n +O(1)$ maximum load in the standard
balls and bins setting.

Our approach uses a standard majorization and coupling argument, where
the coupling links the random choices made by the processes when using double hashing 
and using random hashing while maintaining the fidelity of both individual processes.
(See, e.g., \cite{ABKU,Steger}, or \cite{majorization} for more 
background on majorization.)  Let $\vec{x}=(x_1,\ldots,x_n)$ be a vector
with elements in non-increasing order, so $x_1 \geq x_2 \ldots \geq x_n$,
and similarly for $\vec{y}=(y_1,\ldots,y_n)$.
We say that $\vec{x}$ majorizes $\vec{y}$ if $\sum_{i=1}^n x_i = \sum_{i=1}^n
y_i$ and, for $j < n$, $\sum_{i=1}^j x_i \geq \sum_{i=1}^j y_i$.  For
two Markovian processes $X$ and $Y$, we say that $X$ stochastically
majorizes $Y$ if there is a coupling of the processes $X$ and $Y$ so
that at each step under the coupling the vector representing the state
of $X$ majorizes the vector representing the state of $Y$.  We note that because we
use the loads of the bins as the state, the balls and bins processes
we consider are Markovian.

We make use of the following simple and standard lemma.  (See, for example, \cite[Lemma 3.4]{ABKU}.)
\begin{lemma} 
\label{lem:maj0}
If $\vec{x}$ majorizes $\vec{y}$ for vectors 
$\vec{x}$, $\vec{y}$ of positive integers, and $e_i$ represents a unit vector with
a 1 in the $i$th entry and 0 elsewhere, then $\vec{x}+e_i$ majorizes $\vec{y}+e_j$ for $j \geq i$.  
\end{lemma}

\begin{theorem}
Let process $X$ be the process where $m$ balls are placed into $n$
bins with two distinct random choices, and $Y$ be the corresponding scheme
with $d > 2$ choices using double hashing.  Then $X$ stochastically
majorizes $Y$. 
\end{theorem} 
\begin{proof}
At each time step, we let $\vec{x}(t)$ and $\vec{y}(t)$ be the vectors 
corresponding to the loads sorted in decreasing order.  
We inductively claim that $\vec{x}(t)$ majorizes $\vec{y}(t)$ at all time
steps under the coupling of the processes where if the 
$a$th and $b$th bins in the sorted order for $X$ are chosen,
the $a$th and $b$th bins in the sorted order for $Y$ are chosen
as the first two choices, and then the remaining choices are determined by double hashing.
That is, the $d$ hash choices are such that the gap between successive choices is $b-a$,
so the choices are $a$, $b$, $2b-a$, $3b-2a$, and so on (modulo the size of the table).  
Clearly $\vec{x}(0)$ majorizes $\vec{y}(0)$ as the vectors are equal. 
It is simple to check that this process maintains the majorization
using Lemma~\ref{lem:maj0}, as the coordinate that increases in $\vec{y}(t)$ at each step is deeper in the sorted
order than the coordinate that increases in $\vec{x}(t)$. 
\end{proof}

As two random choices stochastically majorizes $d$ choices from double hashing under this coupling,
we see that 
$$\Pr(x_1 \geq c) \geq \Pr(y_1 \geq c)$$
for any value $c$. 
Since the seminal result of \cite{ABKU} shows that using two
choices gives a maximum load of $\log \log n +O(1)$
with high probability,
we therefore have this corollary.
\begin{corollary}
The maximum load using $d > 2$ choices and double hashing for $n$ balls and $n$ bins is 
$\log \log n +O(1)$ with high probability.
\end{corollary} 

We note that similarly, when using double hashing, we can show that using $d$ choices stochastically
majorizes using $d+1$ choices.

\subsection{A Witness Tree Argument}
\label{sec:witness}
It is well known that $d>2$ choices performs better than 2 choices
for multiple-choice hashing;  while the maximum load remains $O(\log \log n)$,
the constant factor depends on $d$, and can be important in practice.
Our simple majorization argument does not provide this type of bound, so
to achieve it, we next utilize the witness tree approach, following closely the work of
V\"{o}cking \cite{vocking}.  (See also \cite{simplified} for related
arguments.)  While we discuss the case of insertions
only, the arguments also apply in settings with deletions as well; see
\cite{vocking} for more details.  Similarly, here we consider only the
standard balls and bins setting of $n$ balls and $n$ bins with $d \geq
3$ being a constant, but similar results for $m = cn$ balls for some constant $c$ can also be derived
by simply changing the ``base case'' at the leaves of the witness tree accordingly,
and similar results for V\"{o}cking's scheme can be derived by using the ``unbalanced'' witness
tree used by V\"{o}cking \cite{vocking} in place of the balanced one.

These methods allow us to prove statements of the following form:
\begin{theorem} 
\label{thm:vresult}
Suppose $n$ balls are placed into $n$ bins using
the balanced allocation scheme with double hashing as described above.
Then with $d$ choices the maximum load is $\log \log n / \log d + O(d)$
with high probability.
\end{theorem}

We note that, while V\"{o}cking obtains a bound of 
$\log \log n / \log d + O(1)$, we have an $O(d)$ term that appears
necessary to handle the leaves in our witness tree.  (A similar 
issue appears to arise in \cite{Woelfel}.)  For constant $d$ these
are asymptotically the same;  however, an $O(1)$ additive term is
more pleasing both theoretically and potentially in practice.  
How we deviate from V\"{o}cking's argument is explained below.

\begin{proof}
Following \cite{vocking}, we define a witness tree, which is a tree-ordered (multi)set of
balls. Each node in the tree represents a ball, inserted at a certain
time; the $i$th inserted ball corresponds to time $i$ in the natural way.  The
ball represented by the root $r$ is placed at time $t$, and a child
node must have been inserted at a time previous to its parent.  A leaf
node in V\"{o}cking's argument is {\em activated} if each of the $d$
locations of the corresponding ball contains at least three balls when
it is inserted.  An edge $(u,v)$ is activated if when $v$ is the $i$th
child of $u$, then the $i$th location of $u$'s ball is the same as one
of the locations of $v$'s ball.  A witness tree is activated if all of
its leaf nodes and edges are activated.

Following V\"{o}cking's approach, we first bound the probability that
a witness tree is activated for the simpler case where the nodes of the
witness trees represent distinct balls.  The argument then can be
generalized to deal with witness trees where the same ball may appear
multiple times.  As this follows straightforwardly using the technical
approach in \cite{vocking}, we do not provide the full argument here.

We now explain where we must deviate from V\"{o}cking's argument.  The
original argument utilizes the fact that most $n/3$ bins have load at
least 3, deterministically.  As leaf nodes
in V\"{o}cking's argument are required to have all $d$ choices of bins
have load at least 3 to be activated, a leaf node corresponding to a
ball with $d$ choices of bins is activated with probability at most
$3^{-d}$, and a collection of $q$ leaf nodes are all activated with
probability $3^{-dq}$.  However, this argument will not apply in our case, because
the choices of bins are not independent when using double hashing, and
depending on which bins are loaded, we can obtain very different
results.  For example, consider a case where the first $n/3$ bins have
load at least 3.  The fraction of choices using double hashing where
all $d$ bins have load at least 3 is significantly more than $3^{-d}$,
which would be the probability if $n/3$ bins with load 3 were randomly
distributed.  Indeed, for a newly placed ball $j$, if $f(j)$ and
$g(j)$ are both less than $n/(3(d+1))$, all $d$ choices will have load
at least 3, and this occurs with probability at least
$(9(d+1)^2)^{-1}$.  While such a configuration is unlikely, the
deterministic argument used by V\"{o}cking no longer applies.

We modify the argument to deal with this issue.  
In our double hashing setting, let us call a leaf active if either
\vspace{-0.05in}
\begin{itemize}
\item Some ball in the past has two or more of the bins at this leaf among
its $d$ choices.
\vspace{-0.1in}
\item All the $d$ bins chosen by this ball have previously been chosen by $4d$
previous balls.
\end{itemize}
\vspace{-0.05in}
The probability that any previous ball has hit two or more of the bins
at the leaf is $O(d^4n^{-1})$: there are ${d \choose 2}$ pairs of bins
from the $d$ choices at the leaf; at most $d(d-1)$ pairs of positions within the
$d$ choices where that pair could occur in any previous ball; at most $n$ possible previous
balls; and each bad choice that leads that previous ball to have a specific pair of bins in a
specific pair of positions occurs with probability $1/(n(n-1))$.  Once
we exclude this case, we can consider only balls that hit at most one
of the $d$ bins associated with the leaf.

For any time corresponding to a leaf, we bound
the probability that any specific bin has been chosen by $4d$ or more previous balls.
We note by symmetry that the probability any specific ball chooses a specific bin 
is $d/n$.  The probability in question is then at most
$${n \choose 4d}\left (\frac{d}{n} \right)^{4d} \leq
\frac{d^{4d}}{(4d)!} < \left ( \frac{e}{4} \right )^d,$$
which is less than $\frac{1}{3}$ whenever $d \geq 3$.  Further, once we consider the case of previous
balls that choose two or more bins at this leaf separately, the events
that the $d$ bins chosen by this ball have previously been chosen by
$4d$ previous balls are negatively correlated.  Hence, we find the
probability a specific leaf node is activated is less than $3^{-d}$.

However, following \cite{vocking}, we need to consider a collection of $q$ leaves
and show the probability that they are all active is at most $3^{-dq}$.
We will do this below by using Azuma's inequality to show the fraction of choices
of hash values from double hashing that lead to an activated ball is
less than $3^{-d}$ with high probability.  As balls corresponding to
leaves independently choose their hash values, this result suffices.

Let $S$ be the set of pairs of hash values that generate $d$ values
that would activate a leaf at time $n$.  We have $\E[|S|] < \left (
\frac{e}{4} \right )^d n(n-1) + cd^4(n-1)$ for some constant $c$, so
$\E[|S|] > (3^{-d} - \gamma)n(n-1)$ for some constant $\gamma$ and large
enough $n$.  Consider the Doob martingale obtained by revealing the bins
for the balls one at a time.  Each ball can change the final value of $S$
by at most $dn$, since the bin where any ball is placed is involved in 
less than $dn$ choices of pairs.  Azuma's inequality (e.g., \cite[Section 12.5]{MU})
then yields
$$\Pr(|S| > 3^{-d}n(n-1)) \leq \mbox{exp}(-\delta n)$$
for a constant $\delta$ that depends on $d$ and $\gamma$.  
It follows readily that the fraction of pairs of hash values that activate a leaf 
is at most $3^{-d}$ with very high probability throughout the process;  by conditioning
on this event, we can continue with V\"{o}cking's argument.  
(The conditioning only adds an exponentially
small additional probability to the probability the maximum load exceeds
our bound.)

Specifically, we note for there to be a bin of load $L+4d$, there must
be an activated witness tree of depth $L$.  We can bound the
probability that some witness tree (with distinct balls) of depth $L$
is activated.   The probability an edge is activated is the probability a ball
chooses a specific bin, which as previously noted is $d/n$.  As all
balls are distinct, the probability that a witness tree of $m$ balls
has all edges activated is $(d/n)^{m-1}$, and as we have shown the probability
of all leaves being activated is bounded above by $3^{-dq}$ where $q=d^L$ is the number of leaves.
Following \cite{vocking}, as there are at most $n^m$ ways of choosing the balls
for the witness tree, the probability that there exists an active witness tree 
is at most 
\begin{eqnarray*}
n^m \left ( \frac{d}{n} \right)^{m-1} 3^{-dq} & \leq & n \cdot d^{2q} \cdot 3^{-dq} \\
 & \leq & n \cdot 2^{-q} \\
 & = & n \cdot 2^{-d^{L}}.
\end{eqnarray*}
Hence choosing $L \leq \log_d \log_2 n + \log_d(1+\alpha)$ guarantees a maximum load
of $L + 4d$ with probability $O(n^{-\alpha})$. 
\end{proof}

\section{The Fluid Limit Argument}
\label{sec:fluid}

We now consider the fluid limit approach of \cite{MitzenmacherThesis}.
(A useful survey of this approach appears in \cite{Diaz}.)
The fluid limit approach gives equations that describe the asymptotic
fraction of bins with each possible integer load, and concentration
around these values follows from martingale bounds (e.g.,
\cite{EK,Kurtz,Wormald}).  Values can easily be determined
numerically, and prove highly accurate even for small numbers of balls
and bins.  We show that the same equations apply even in the setting
of double hashing, giving a theoretical justification for our
empirical findings in Appendix~\ref{sec:sims}.  This approach can be
easily extended to other multiple choice processes (such as V\"{o}cking's
scheme and the queuing setting).  We emphasize that the
fluid limit approach does not, in itself, yield bounds of
the type that the maximum load is $O(\log \log n)$ with high
probability naturally; rather, it says that for any constant integer
$i$, the fraction of bins of load $i$ is concentrated around the value
obtained by the fluid limit.  One generally has to do additional work --
generally similar in nature to the arguments in the proceeding sections -- to
obtain $O(\log \log n)$ bounds.  As we already have an $O(\log \log n)$ bound from
alternative techniques, here our focus is on showing the fluid limits are the same
under double hashing and fully random hashing, which explains our empirical findings.  
(We show one could achieve an $O(\log \log n)$ bound from the results of this section
 -- actually bound of $\log_d \log_2 n + O(1)$ -- in Appendix~\ref{sec:followon}.)

The standard balls and bins fluid limit argument 
runs as follows.  Let $X_i(t)$ be a random variable denoting the number of 
bins with load {\em at least} $i$ after $tn$ balls have been thrown;  hence $X_0(0) =n$
and $X_i(0) = 0$ for all $i \geq 1$.  Let $x_i(t) = X_i(t)/n$.  For $X_i$ to increase when
a ball is thrown, all of its choices must have load at least $i-1$, but not all of them
can have load at least $i$.  Hence for $i \geq 1$
$$\E[X_i(t + 1/n) - X_i(t)] = (x_{i-1}(t))^d - (x_{i}(t))^d.$$
Let $\Delta(x_i) = x_i(t + 1/n) - x_i(t)$ and $\Delta(t) = 1/n$.  Then the above can be written as:
$$\E \left [ \frac{\Delta(x_i)}{\Delta(t)} \right ] = (x_{i-1}(t))^d - (x_{i}(t))^d.$$
In the limit as $n$ grows, we can view the limiting version of the above equation as
$$\frac{dx_i}{dt} = x_{i-1}^d - x_{i}^d,$$
where we remove the $t$ on the right hand side as the meaning is clear.  
Again, previous works \cite{EK,Kurtz,Wormald} justify how the Markovian load balancing process
converges to the solution of the differential equations.\footnote{In particular, the technical conditions
corresponding to Wormald's result \cite[Theorem 1]{Wormald} hold, and this theorem gives the appropriate convergence;
we explain further in our Theorem~\ref{mainthm}.}  
Specifically, it follows from Wormald's theorem \cite[Theorem 1]{Wormald} that
$$X_i(t) = nx_i(t) + o(n)$$
with probability $1-o(1)$, 
or equivalently that the fraction of balls of load $i$ is within $o(1)$ of the 
result of the limiting differential equations with probability $1-o(1)$.
These equations allow us to compute
the limiting fraction of bins of each load numerically, and these results closely match our 
simulations, as for example shown in Table~\ref{table_6}.  

\begin{table*}[thp]
\centering
\begin{tabular}{|c|c|c|c|} \hline
Tail load & Fluid Limit & Fully Random & Double Hashing \\ \hline
$\geq 1$  & 0.8231 & 0.8231 & 0.8231 \\
$\geq 2$  & 0.1765 & 0.1764 & 0.1764 \\
$\geq 3$  & 0.00051 & 0.00051 & 0.00051 \\ \hline
\end{tabular}
\caption{3 choices, fluid limit ($n=\infty$) vs. $n=2^{14}$ balls and bins}
\label{table_6}
\end{table*}

Given our empirical results, it is natural to conclude that these
differential equations must also necessarily describe the behavior of
the process when we use double hashing in place of standard hashing.
The question is how can we justify this, as the equations were derived
utilizing the independence of choices, which is not the case for double hashing.

We now prove that, for constant number of choices $d$, 
constant load values $i$, and a constant time $T$ (corresponding to $Tn$ total balls), 
the loads of the bins chosen by double hashing behave essentially the same as though
the choices were independent,
in that, with high probability over the entire course of the process, 
$$\E[X_i(t + 1/n) - X_i(t)] = (x_{i-1}(t))^d - (x_{i}(t))^d +o(1);$$
that is, the gap is only in $o(1)$ terms.  
This suffices for \cite[Theorem 1]{Wormald} (specifically
condition (ii) of \cite[Theorem 1]{Wormald} allows such $o(1)$ differences).
The result is that double hashing has no effect on the fluid limit analysis.
(Again, we emphasize our restriction to constant choices $d$, constant load values $i$, 
and constant time parameter $T$.)  
Our approach is inspired by the work of Bramson, Lue, and Prabhakar
\cite{Bramson}, who use a similar approach to obtain asymptotic
independence results in the queueing setting.  However, there the concern
was on limiting independence in equilibrium with general service time
distributions, and the choices of queues were assumed to
be purely random. We show that this methodology can be applied to the double
hashing setting.

\begin{lemma}  
\label{lem:mainlemma}
When using double hashing, with high probability over the entire course of the process, 
$$\E[X_i(t + 1/n) - X_i(t)] = (x_{i-1}(t))^d - (x_{i}(t))^d +o(1).$$
\end{lemma}  
\begin{proof}  
We refer to the {\em ancestry list} of a bin $b$ at time $t$ as
follows.  The list begins with the balls $z_1,z_2,\ldots,z_{g(b,t)}$
that have had bin $b$ as one of their choices, where $g(b,t)$ is
the number of balls that have chosen bin $b$ up to time $t$.  Note
that each $z_i$ is associated with a corresponding time $t_i$ and
$d-1$ other bin choices.  For each $z_i$, we recursively add the list
of balls that have chosen each of those $d-1$ bins up to time $t_i$, and so on recursively.  
We also think of the bins associated with these balls as being part of the ancestry
list, where the meaning is clear.  It is
clear that the ancestry list gives all the necessary information to
determine the load of the bin $b$ at time $t$ (assuming the
information regarding choices is presented in such a way to include
how placement will occur in case of ties; e.g., the bin choices are
ordered by priority).  We note that the ancestry list holds more information
(and more balls and bins) than the witness trees used by V\"{o}cking (and by us
in Section~\ref{sec:witness}).

In what follows below let us assume $n$ is prime for convenience (we
explain the difference if $n$ is not prime in footnotes).  We claim
that for asymptotic independence of the load among a collection of $d$
bins at a specific time when a new ball is placed, it suffices to show
that these ancestry lists are small.  Specifically, we start with
showing in Lemma~\ref{lem:branching} that all ancestry lists contain
only $O(\log n)$ associated bins with high probability.  We then show
as a consequence in Lemma~\ref{lem:small} that the ancestry lists of
the bins associated with a newly placed ball have no bins in common
with high probability.  This last fact allows us to complete the main
lemma, Lemma~\ref{lem:mainlemma}.

\begin{lemma}
\label{lem:branching}
The number of bins in the ancestry list of every bin after the first $Tn$ steps is at most 
$O(\log n)$ with high probability.
\end{lemma}
\begin{proof}
We view the growth of the ancestry list as a variation of the standard branching process, by going
backward in time.  Let $B_{0} = 1$ correspond to size of an initial ancestry list of a bin $b$, consisting
of the bin itself.  If the $(Tn)$th ball thrown has $b$ as one of its $d$ choices, then $d-1$ additional bins
are added to the ancestry list, and we then have $B_{1} = d$;  otherwise we have no change and $B_{1} = 1$.   
(Note that when measuring the size of the ancestry list in bins, each bin is counted only once,
even if it is associated with multiple balls.)  
If the $(Tn-1)$st ball thrown has a bin in the ancestry list as one of its $d$ choices, then (at most) $d-1$
bins are added to the ancestry list, and we set $B_2 = B_1 + d-1$;  otherwise, we have $B_2 = B_1$.  We
continue to add to the ancestry list with at each step $B_i = B_{i-1} + d-1$ or $B_i = B_{i-1}$, depending
on whether the $(Tn-i+1)$st ball has one of it choices as a bin on the ancestry list, or not.  

This process is {\em almost} equivalent to a Galton-Watson branching process where
in each generation, each existing element produces 1 offspring with
probability $1-d/n$ (or equivalently, moves itself into the next
generation), or produces $d$ offspring (adding $d-1$ new elements)
with probability $d/n$.  The one issue is that the production of
offspring are not independent events; at most $d-1$ elements are added
at each step in the process.  (There is also the issue that perhaps
fewer than $d-1$ elements are added when elements are added to the
ancestry list; for our purposes, it is pessimistic to assume $d-1$
offspring are produced.)  Without this dependence concern, standard
results on branching process would give that $E[B_{Tn}] =
(1+d(d-1)/n)^{Tn} \leq e^{Td(d-1)}$, which is a constant.  Further, we could apply
(Chernoff-like) tail bounds from Karp and Zhang \cite[Theorem 1]{KZ}, which states the following:
for a supercritical finite time branching process $\left \{ Z_n \right \}$ over $n$ time steps starting with $Z_0 =1$,
with mean offspring per element $\E[Z_1] = \rho >1$, and with $\E[e^{Z_1}] < \infty$, 
there exists constants
$c_1$ and $c_2$ such that
$$\Pr(Z_n > \gamma \rho^n) < c_1 e^{-c_2 \gamma}.$$
In our setting, that would give that
there exists constants
$c_1$ and $c_2$ such that
$$\Pr(B_{Tn} > \gamma (1+d(d-1)/n)^{Tn} ) < c_1 e^{-c_2 \gamma}.$$
This would give our desired $O(\log n)$ high probability bound on the size of the ancestry list.

To deal with this small deviation, it suffices to consider a modified Galton-Watson process
where each element produces $d$ offspring with probability $d'/n$;  we shall see that $d'= d+1$
suffices.  Let $B'$ be the resulting size of this Galton Walton process.  From the above
we have that $B' < c \log n$ with high probability for some suitable constant $c$.    

Our original desired ancestry list process is dominated by a process where $B_i =  \min(B_{i-1} + d-1,n)$ with probability $\min(B_{i-1}d/n,1)$ and $B_i = B_{i-1}$ otherwise, and this process is in turn dominated for values of $B_i$
up to $c \log n$ by a Galton-Waston branching process  where the constant $d'$ satisfies
$$1-(1-d'/n)^x \geq dx/n$$
for all $1 \leq x \leq c \log n$, so that at every stage the Galton-Watson process is more likely to have at least $d-1$ new offspring (and may have more).  
We see $d' = d+1$ suffices, as $$1-(1-(d+1)/n)^x = x(d+1)/n - O(dx^2/n^2)$$
which is greater than $dx/n$ for $n$ sufficiently large when $x$ is $O(\log n)$.  
The straightforward step by step coupling of the processes yields that 
$$\Pr(B_{Tn} > c \log n) \leq \Pr(B' > c \log n),$$ giving our desired bound. 
 

We also suggest a slightly cleaner alternative, which may prove useful
for other variations: embed the branching process in a continuous time
branching process.  We scale time so that balls are thrown as a
Poisson process or rate $n$ per unit time over $T$ time units.  Each element therefore
generates $d-1$ new offspring at time instants that are exponentially
distributed with mean $1/d$ (the average time before a ball hits any bin
on the ancestry list).  Again, assuming $d-1$ new offspring is a pessimistic
bound. If we let $C_t$ be the number of elements at time $t$ (starting from 1 element at time 0),
it is well known (see, e.g., \cite[p.108 eq. (4)]{AN}, and note that generating $d-1$
new offspring is equivalent to ``dieing'' and generating $d$ offspring) that for such a process, 
$$\E[C_t] = e^{td(d-1)}.$$
In our case, we run to a fixed time $T$ and $\E[C_T] = e^{Td(d-1)}$, a constant.
Indeed, in this specific case, the generating function for the distribution
of the number of elements is known (see, e.g., \cite[p.109]{AN}), allowing us to directly apply a Chernoff bound.
Specifically, 
$$\E[s^{C_t}] = se^{-dt}[1-(1-e^{-d(d-1)t})s^{d-1}]^{-1/(d-1)}.$$
Hence we have
\begin{eqnarray*}
\Pr(C_T > \gamma e^{Td(d-1)}) & = & \Pr(e^{C_T} > e^{\gamma e^{Td(d-1)}}) \\
& \leq & e^{-\gamma e^{Td(d-1)}} \E[e^{C_T}] \\
& \leq & c_3 e^{-c_4 \gamma}
\end{eqnarray*}
for constants $c_3$ and $c_4$ that depend on $d$ and $T$.  Hence, this gives that
the size of the ancestry list as viewed from the setting of the continuous branching
process is $O(\log n)$ with high probability.  

The last concern is that running the continuous process for time $Tn$ does not guarantee that 
$Tn$ balls are thrown; this can be dealt with by thinking of the process running
for a slightly longer time $T' > T$.  That is, choose $T'= T +\epsilon$ for a
small constant $\epsilon$.  Standard Chernoff bounds on the
Poisson random variables then guarantee that at least $Tn$ balls are then thrown with high probability,
and the size of the ancestry lists are stochastically monotonically increasing with the number
of balls thrown.  Changing to $T'$ time units maintains that each ancestry list is $O(\log n)$
with high probability.  

Finally, by choosing the constant in the $O(\log n)$ term appropriately, we
can achieve a high enough probability to apply a union bound so that this holds for all
ancestry lists simultaneously with high probability.
\end{proof}

We now use Lemma~\ref{lem:branching} to show the following.
\begin{lemma}
\label{lem:small}
The bins in the ancestry lists of the $d$ choices are disjoint with probability 
$1-\eta$ for $\eta = O(d^2 \log^2 n/n) = o(1)$.
\end{lemma}
\begin{proof} Let {\cal F} be the probability that the bins are disjoint, and
let ${\cal E}$ be the event that no pair of the $d$ choices were previously
chosen by the same ball.  If ${\cal E}$ occurs, the ancestry lists are clearly not disjoint.
Hence we wish to bound 
$$\Pr(F) \leq \Pr({\cal E}) + \Pr({\cal F} | \neg{\cal E}).$$

Consider any two of the $d$ bins chosen by the ball being placed.
Each of the up to $Tn$ previous balls have $O(d^2)$ ways of choosing those two bins
as two of their $d$ choices (e.g., picking that bin as the 2nd and 4th
choice, for example), and the probability of choosing those two bins
for each possible pair of choice positions is $O(1/n^2)$.\footnote{If
$n$ is not prime, this probability is $O(1/n\phi(n))$, where $\phi$
is the Euler totient function counting the number of numbers less
than $n$ that are relatively prime to $n$.  We note $\phi(n)$ is
usually $\Omega(n)$ and is always $\Omega(n/\log \log n)$, so this
does not affect our argument substantially.}  There are ${d \choose 2}$
pairs of balls, so by a union bound
$\Pr({\cal E})$ is $O(Td^4/n^2)$.  

Now suppose that no pair of the $d$ bins were previously chosen by the
same ball.  Suppose the bins
for each of the ancestry lists of the $d$ choices are
ordered in some fixed fashion (say according to decreasing ball time, randomly
permuted for each ball).  We consider the probability that the $i$th
bin in the ancestry list of one bin matches the $j$th bin in another. 
Since the lists do not share any ball in common,
the $j$th bin in the second list matches the $i$th bin in the first list 
with probability only $O(1/n)$, as even conditioned on the value of the $i$th bin
on the first list, the $j$th bin on the second list is uniform over $\Omega(n)$ possibilities.\footnote{Again,
for $n$ not prime, we may use $\Omega(\phi(n))$ possibilities.}
We now condition on all of the $d$ ancestry lists being of size $O(\log n)$;  from
Lemma~\ref{lem:branching}, this can be made to occur with any inverse polynomial probability
by choosing the constant factor in the $O(\log n)$ term, so we assume this bound on 
ancestry list sizes.  In his case,
the probability of a match among any of the $d$ bins is only $O(d^2 \log^2 n/n)$ in total, 
where the $d^2$ factor is from the $d \choose 2$ possible ways of choosing bins, and the $\log^2 n$
term follows the bound on the size ancestry lists.  Hence 
$\Pr({\cal F} | \neg{\cal E})$ is $O(d^2 \log^2 n/n)$, and
the total probability
that the ancestry lists of the $d$ choices are {\em not} disjoint is $\eta = O(d^2 \log^2 n/n) = o(1)$.
\end{proof}

We now show that this yields the Lemma~\ref{lem:mainlemma}.  
To clarify this, consider bins $b_1,b_2,\ldots,b_d$ that were chosen by a
ball at some time $t+1/n$.  (Recall our scaling of time.)  The
probability that all $d$ bins have load at least $i$ at that time is
equivalent to the probability that each bin $b_j$ has
a corresponding ancestry list $A_j$ showing that it has load $i$ at some time
$u_j \leq t$.  Fix a collection of ancestry lists $A_j$, and let 
$E_j$ be the event defined by ``bin $b_j$ has ancestry list $A_j$''.  
If these ancestry lists have disjoint sets of bins,
then the corresponding balls in each ancestry list occur at different times
and have no intersecting bins, and as such
$$ \Pr \left (  \cap_j E_j \right ) = \prod_j \Pr(E_j).$$
For constant $i$, $t$, and $d$, the probability that all $d$ bins have load at 
least $i$ is constant.  Hence, if the probability that the ancestry lists for
the $d$ bins intersect at any bin is $\eta = o(1)$, we have asymptotic independence.
Specifically, let $\cal X$ be the set of collections of $d$ ancestry lists 
for balls $b_1,b_2,\ldots,b_d$ that yield that each bin has load at least $i$ at time $t$, let $\cal Y$ be the subset of collections in $\cal X$ where the $d$ ancestry lists have no bins
in common, and for a collection $Z$ in $\cal X$ let $E_j(Z)$ be the corresponding
event defined by ``bin $b_j$ has ancestry list $A_j$ in collection $Z$''.  
Then
\begin{eqnarray*}
\sum_{Z \in {\cal X}} \Pr \left (  \cap_j E_j(Z) \right ) &= & \left [ \sum_{Z \in {\cal Y}} \Pr \left (  \cap_j E_j(Z) \right ) \right ] + o(1) \\
& = & \sum_{Z \in {\cal Y}} \left (\prod_j \Pr E_j(Z) \right ) + o(1) \\
& = & \sum_{Z \in {\cal X}} \left (\prod_j \Pr E_j(Z) \right ) + o(1). 
\end{eqnarray*}
Here the first line uses that the $d$ ancestry lists intersect somewhere with probability
$o(1)$;  the second lines uses that for ancestry lists in  $\cal Y$ we probability of the
intersection is the product of the probabilities;  and the third line is again because the
the collections $Z$ in ${\cal X} - {\cal Y}$ have total probability $o(1)$.
Hence up to an $o(1)$ term, the behavior is the same as if the $d$ choices were independent
(with respect to all bins having load at least $i$).  
Thus $$\E[X_i(t + 1/n) - X_i(t)] = (x_{i-1}(t))^d - (x_{i}(t))^d +o(1)$$
as needed.
\end{proof}

As a result of Lemma~\ref{lem:mainlemma}, we have the following theorem, generalizing the differential equations
approach for balanced allocations to the setting of double hashing.
\begin{theorem}
\label{mainthm}
Let $i$, $d$, and $T$ be constants.  Suppose $Tn$ balls are
sequentially thrown into $n$ bins with each ball having $d$ choices
obtained from double hashing and each ball being placed in the least
loaded bin (ties broken randomly).  Let $X_i(T)$ be the number of
bins of load at least $i$ after the balls are thrown.  Let $x_i(t)$
be determined by the family of differential equations 
$$\frac{dx_i}{dt} = x_{i-1}^d - x_{i}^d,$$
where $x_0(t) = 1$ for all time and $x_i(0) = 0$ for $i \geq 1$.  
Then with probability $1-o(1)$, 
$$\frac{X_i(T)}{n} = x_i(T) + o(1).$$
\end{theorem}
\begin{proof}
This follows from the fact that $$\E[X_i(t + 1/n) - X_i(t)] = (x_{i-1}(t))^d - (x_{i}(t))^d +o(1),$$
and applying Wormald's result \cite[Theorem 1]{Wormald}.

We remark that Theorem 1 of \cite{Wormald} includes other technical
conditions that we briefly consider here.  The first condition is that $|X_i(t + 1/n) - X_i(t)|$ is bounded by a constant;  
all such values here are bounded by 1.  The second (and only challenging) condition exactly corresponds to our statement that
$\E[X_i(t + 1/n) - X_i(t)] = (x_{i-1}(t))^d - (x_{i}(t))^d +o(1)$ over the course of the process.  The third condition is
our functions on the right hand side, that is $(x_{i-1}(t))^d - (x_{i}(t))^d$, are continuous and satisfy a Lipschitz
condition on an open neighborhood containing the path of the process.  These functions are continuous on the domain
where all $x_i \in [0,1]$ up to the value $i$ being considered, and they satisfy
the Lipschitz condition as
\begin{eqnarray*}
|(x_{i-1}(t))^d - (x_{i}(t))^d| \! \! & \leq & \! \! |x_{i-1}(t) - x_{i}(t)| \sum_{j=0}^{d-1} (x_{i-1}(t))^{j}(x_i(t))^{d-1-j} \\
\! \! & \leq & \! \! d|(x_{i-1}(t)) - (x_{i}(t))|,
\end{eqnarray*}
taking note that all $x_i,x_{i-1}$ values are in the interval $[0,1]$.  Hence the conditions for Wormald's theorem are met.
\end{proof}

The following corollary, based on the known fact that the result of
Theorem~\ref{mainthm} also holds in the setting of fully random
hashing \cite{MitzenmacherThesis}, states that the difference between
fully random hashing and double hashing is vanishing.
\begin{corollary}
\label{cormain}
Let $i$, $d$, and $T$ be constants.  Consider two processes, where in each $Tn$ balls are
sequentially thrown into $n$ bins with each ball having $d$ choices
and each ball being placed in the least loaded bin (ties broken randomly),
In one process, the $d$ choices are fully random;  in the other, the $d$ choices 
are made by double hashing.  Then with probability $1-o(1)$, the fraction of bins
with load $i$ differ by an $o(1)$ additive term.
\end{corollary}

Given the results for the differential equations, it is perhaps unsurprising that one can use these methods
to obtain, for example, a maximum load of $\log \log n/ \log d+ O(1)$ maximum load for $n$ balls in $n$ bins,
using the related layered induction approach of \cite{ABKU}.  While we suggest this is not the main point (given Theorem~\ref{thm:vresult}), we provide further details in Appendix~\ref{sec:followon}.

\section{Conclusion}
We have first demonstrated empirically that using double hashing with
balanced allocation processes (e.g., the power of (more than) two
choices), surprisingly, does not noticeably change performance when
compared with fully random hashing.  We have then shown that previous
methods can readily provide $O(\log \log n)$ bounds for this approach.
However, explaining why the fraction of bins of load $k$ for each $k$
appears the same requires revisiting the fluid limit model for such
processes.  We have shown, interestingly, that the same family of
differential equations applies for the limiting process.  Our argument
should extend naturally to other similar processes;  for example, 
the analysis can similarly be made to apply in a straightforward fashion for the 
differential equations for V\"{o}cking's $d$-left scheme \cite{MV}.

This opens the door to the interesting possibility that double hashing
can be suitable for other problem or analyses where this type of fluid
limit analysis applies, such as low-density parity-check codes
\cite{LMSS}.  Here, however, the asymptotic independence required
was aided by the fact that we were looking at the history of the process,
allowing us to tie the ancestry lists to a corresponding branching process.
Whether similar asymptotic independence can be derived for other problems remains
to be seen.  For other problems, such as cuckoo hashing, the
fluid limit analysis, while an important step, may not offer a
complete analysis.  Even for load balancing problems, fluid limits do not
straightforwardly apply for the heavily loaded case where the number of 
balls is superlinear in the number of bins \cite{Steger}, and it is unclear
how double hashing performs in that setting.  So again, determining more generally where double
hashing can be used in place of fully random hashing without
significantly changing performance may offer challenging future questions.

\section*{Acknowledgments}  The author thanks George Varghese for the discussions which led to the
formulation of this problem, and thanks Justin
Thaler for both helpful conversations and offering several suggestions for improving the presentation of results.  

\bibliographystyle{plain}

\appendix

\section{Empirical Results}
\label{sec:sims}

We have done extensive simulations to test whether using double hashing
in place of idealized random hashing makes a difference for several
multiple choice schemes.  Theoretically, of course, there is some
difference; for example, the probability that $k$ balls choose the
same specified set of $d$ bins is $O(n^{-dk})$ with fully random choices, and only
$O(n^{-2k})$ with double hashing (where the order notation may hide
factors that depend on $d$).  Hence, to be clear, the best we can hope
for are differences up to $o(1)$ events.  Empirically, however, our
experiments suggest the effects on the distribution of the loads, or
in particular on the probability the maximum load exceeds some value,
are all found deeply in the lower order terms. Experiments show that
unless especially rare events are of special concern, we expect the
two to perform similarly.  

\subsection{The Standard $d$-Choice Scheme}

We first consider $n$ balls and bins using $d$ choices without
replacement, comparing fully random choices with
double hashing.\footnote{We also considered $d$ choices with
replacement, but the difference was not apparent except for very
small $n$, so we present only results without replacement.  However, we
note that conversations with George Varghese regarding hardware
settings with small $n$ originally motivated our examination of this
approach.}  When using double hashing we choose an odd stride value
as explained previously.  All results presented are over 10000 trials.
Table~\ref{table_1} shows the distributions
of bin loads for 3 and 4 choices, averaged over all 10000 trials, for
$n=2^{16}$ and $n=2^{18}$.  (Recall $n = 2^{14}$ was shown in
Table~\ref{table_example}.)  As can be seen, the deviations are all
very small, within standard sampling error.

\begin{table*}[thp]
\begin{subfigure}[h]{0.45\textwidth}
\centering
\begin{tabular}{|c|c|c|} \hline
Load & Fully Random & Double Hashing \\ \hline
0 & 0.17695 & 0.17693 \\
1 & 0.64661 & 0.64664 \\
2 & 0.17593 & 0.17592 \\
3 & 0.00051 & 0.00051 \\ \hline
\end{tabular}
\caption{3 choices, $n=2^{16}$ balls and bins}
\end{subfigure}
\qquad
\begin{subfigure}[h]{0.45\textwidth}
\centering
\begin{tabular}{|c|c|c|} \hline
Load & Fully Random & Double Hashing \\ \hline
0 & 0.14081 & 0.14083 \\
1 & 0.71841 & 0.71835 \\
2 & 0.14076 & 0.14079 \\
3 & $2.32\cdot 10^{-5}$ & $2.30\cdot 10^{-5}$ \\ \hline
\end{tabular}
\caption{4 choices, $n=2^{16}$ balls and bins}
\end{subfigure}
\\
\begin{subfigure}[h]{0.45\textwidth}
\centering
\begin{tabular}{|c|c|c|} \hline
Load & Fully Random & Double Hashing \\ \hline
0 & 0.17696 & 0.17696 \\
1 & 0.64658 & 0.64648 \\
2 & 0.17595 & 0.17595 \\
3 & 0.00051 & 0.00051 \\ \hline
\end{tabular}
\smallskip
\caption{3 choices, $n=2^{18}$ balls and bins}
\end{subfigure}
\qquad
\begin{subfigure}[h]{0.45\textwidth}
\centering
\begin{tabular}{|c|c|c|} \hline
Load & Fully Random & Double Hashing \\ \hline
0 & 0.14083 & 0.14082 \\
1 & 0.71837 & 0.71838 \\
2 & 0.14078 & 0.14078 \\
3 & $2.31\cdot 10^{-5}$ & $2.32\cdot 10^{-5}$ \\ \hline
\end{tabular}
\caption{4 choices, $n=2^{18}$ balls and bins}
\end{subfigure}
\caption{Essentially indistinguishable differences in simulation between double hashing and fully random hashing. \label{table_1}}
\end{table*}

\begin{table*}[thp]
\begin{subfigure}[h]{0.45\textwidth}
\centering
\begin{tabular}{|c|c|c|} \hline
$n$ & Fully Random & Double Hashing \\ \hline
$2^{10}$ & 39.78 & 39.40 \\
$2^{11}$ & 64.71 & 65.15 \\
$2^{12}$ & 86.90 & 87.05 \\
$2^{13}$ & 98.37 & 98.63 \\ 
$2^{14}$ & 100.00 & 99.99 \\ 
$2^{15}$ & 100.00 & 100.00 \\ \hline
\end{tabular}
\caption{3 choices, fraction with maximum load 3}
\end{subfigure}
\qquad
\begin{subfigure}[h]{0.45\textwidth}
\centering
\begin{tabular}{|c|c|c|} \hline
$n$ & Fully Random & Double Hashing \\ \hline
$2^{10}$ & 2.24 & 2.23 \\
$2^{12}$ & 8.91 & 8.52 \\
$2^{14}$ & 30.75 & 31.42 \\
$2^{16}$ & 78.23 & 77.72 \\ 
$2^{18}$ & 99.77 & 99.79 \\ 
$2^{20}$ & 100.00 & 100.00 \\ \hline
\end{tabular}
\caption{4 choices, fraction with maximum load 3}
\end{subfigure}
\caption{Comparing maximum loads.  The fraction of runs with maximum load 3 is similar. \label{table_2}
}
\end{table*}

\begin{table*}[thp]
\begin{subfigure}[h]{0.45\textwidth}
\centering
\begin{tabular}{|c|c|c|c|c|} \hline
Load & min & avg & max & std.dev. \\ \hline
0 & 36522 & 36913.75 & 37308 & 111.06 \\
1 & 187533 & 188322.55 & 189103 & 222.02 \\
2 & 36516 & 36901.67 & 37298 & 110.96 \\
3 & 1 & 6.04 & 17 & 2.42 \\ \hline
\end{tabular}
\caption{Fully random, load distribution over 10000 trials}
\end{subfigure}
\qquad
\begin{subfigure}[h]{0.45\textwidth}
\centering
\begin{tabular}{|c|c|c|c|c|} \hline
Load & min & avg & max & std.dev. \\ \hline
0 & 36535 & 36916.57 & 37301 & 109.89 \\
1 & 187544 & 188316.93 & 189078 & 219.71 \\
2 & 36524 & 36904.45 & 37297 & 109.85 \\
3 & 1 & 6.06 & 18 & 2.44 \\ \hline
\end{tabular}
\caption{Double hashing, load distribution over 10000 trials}
\end{subfigure}
\caption{Viewing the sample standard deviation, 4 choices, $2^{18}$ balls and $2^{18}$ bins. \label{table_2ba}
}
\end{table*}

\begin{table*}[thp]
\begin{subfigure}[h]{0.45\textwidth}
\centering
\begin{tabular}{|c|c|c|} \hline
Load & Fully Random & Double Hashing \\ \hline
9 &  $6.10 \cdot 10^{-9}$ & $6.10 \cdot 10^{-9}$ \\
10 & $1.28 \cdot 10^{-7}$ & $1.71 \cdot 10^{-7}$ \\
11 & $2.50 \cdot 10^{-6}$ & $2.95 \cdot 10^{-6}$ \\
12 & $4.54 \cdot 10^{-5}$ & $4.51 \cdot 10^{-5}$ \\
13 & 0.00076    & 0.00076    \\
14 & 0.01254    & 0.01254    \\
15 & 0.16885    & 0.16877    \\
16 & 0.62220    & 0.62234    \\
17 & 0.19482    & 0.19475    \\
18 & 0.00079    & 0.00079    \\ \hline
\end{tabular}
\caption{3 choices, $2^{18}$ balls and $2^{14}$ bins}
\end{subfigure}
\qquad
\begin{subfigure}[h]{0.45\textwidth}
\centering
\begin{tabular}{|c|c|c|} \hline
Load & Fully Random & Double Hashing \\ \hline
11 & $2.44 \cdot 10^{-8}$ &  $2.44 \cdot 10^{-8}$ \\
12 & $1.48 \cdot 10^{-6}$ &  $1.34 \cdot 10^{-6}$ \\
13 & $6.92 \cdot 10^{-5}$ &  $6.98 \cdot 10^{-4}$ \\
14 & 0.00349    &  0.00349    \\
15 & 0.13908    &  0.13906    \\
16 & 0.71110    &  0.71114    \\
17 & 0.14622    &  0.14620    \\
18 & $2.86 \cdot 10^{-5}$ &  $2.85 \cdot 10^{-5}$ \\ \hline
\end{tabular}
\caption{4 choices, $2^{18}$ balls and $2^{14}$ bins}
\end{subfigure}
\caption{The similarity in performance persists under higher loads. \label{table_3}
}
\end{table*}

\begin{table*}[thp]
\begin{subfigure}[h]{0.45\textwidth}
\centering
\begin{tabular}{|c|c|c|} \hline
Load & Fully Random & Double Hashing \\ \hline
0 & 0.12420 & 0.12421 \\
1 & 0.75160 & 0.75158 \\
2 & 0.12420 & 0.12421 \\ \hline
\end{tabular}
\caption{4 choices, $2^{14}$ balls and bins}
\end{subfigure}
\qquad
\begin{subfigure}[h]{0.45\textwidth}
\centering
\begin{tabular}{|c|c|c|} \hline
Load & Fully Random & Double Hashing \\ \hline
0 & 0.12421 & 0.12421 \\
1 & 0.75159 & 0.75158 \\
2 & 0.12421 & 0.12421 \\
3 &         & $7.63 \cdot 10^{-10}$ \\ \hline
\end{tabular}
\caption{4 choices, $2^{18}$ balls and bins}
\end{subfigure}
%
\caption{Double hashing performance with V\"{o}cking's $d$-left scheme. \label{table_4}
}
\end{table*}

We may also consider the maximum load.  In Table~\ref{table_2}, we
consider values of $n$ where the maximum load is at most 3, and examine the
fraction of time a load of 3 is achieved over the 10000 trials.
Again, the difference between the two schemes appears small, to the
point where it would be a challenge to differentiate between the two
approaches.

We focus in on the case of 4 choices with $2^{18}$ balls and bins to
examine the sample standard deviation (across {10000} trials) in
Table~\ref{table_2ba}.  This example is representative of behavior in
our other experiments.  By looking at the number of bins of each load
over several trials, we see the sample standard deviation is very
small compared to the number of bins of a given load, whether using
double hashing or fully random hashing, and again performance is
similar for both.

A reasonable question is whether the same behavior occurs if the average load
is larger than 1.  We have tested this for several cases, and again found
that empirically the behavior is essentially indistinguishable.  As an example, 
Table~\ref{table_3} gives results in the case of $2^{18}$ balls being thrown into $2^{14}$ bins, for
an average load of 16.  Again, the differences are at the level of sampling deviations.

\begin{table*}[thp]
\centering
\begin{tabular}{|c|c|c|c|} \hline
$\lambda$ & Choices & Fully Random & Double Hashing \\ \hline
0.9 & 3 & 2.02805 & 2.02813 \\
0.9 & 4 & 1.77788 & 1.77792 \\
0.99 & 3 & 3.85967 & 3.86073 \\
0.99& 4 & 3.24347 & 3.24410 \\ \hline
\end{tabular}
\caption{$n=2^{14}$ queues, average time}
\label{table_5}
\end{table*}

We note that we obtain similar results under variations of the
standard $d$-choice scheme.  For example, using V\"{o}cking's approach
of splitting in $d$ subtables and breaking ties to the left, we obtain
essentially indistinguishable load distributions with fully random
hashing and double hashing.  Table~\ref{table_4} shows results from a
representative  case where $d=4$, again averaging over 10000 trials.  The
case of $n = 2^{18}$ is instructive; this appears very close to the
threshold where bins with load 3 can appear.  While there appears to
be a deviation, with double hashing have some small fraction of bins
with load 3, this corresponds to exactly 2 bins over the 10000 trials.
Further simulations suggest that this apparent gap is less significant
than it might appear; over 100000 trials, for random, the maximum load
was 3 for three trials, while for double hashing, it was 3 for four
trials.

In the standard queueing setting, balls arrive as a Poisson process of
rate $\lambda n$ for $\lambda < 1$ to a bank of $n$ first-in first-out
queues, and have exponentially distributed service times with mean 1.
Jobs are placed by choosing $d$ queues and going to the queue with the
fewest jobs.  The asymptotic equilibrium distributions for such
systems with independent, uniform choices can be found by fluid
limit models \cite{Mitzenmacher,Vvedenskaya}.  We ran 100 simulations of 10000 seconds,
recording the average time over all packets after time 1000 (allowing
the system to ``burn in''.)  An example appears in
Table~\ref{table_5}.  While double hashing performs slightly worse in
these trials, the gap is far less than 0.1\% in all cases.

\section{Extending the Fluid Limit}
\label{sec:followon}

We sketch an approach to extend the fluid limit result to provide an
$O(\log \log n)$ result.  In fact, we show here that for $n$ balls
being thrown into $n$ bins via double hashing, we obtain a load of
$\log \log n / \log d + O(1)$, avoiding the $O(d)$ term of
Section~\ref{sec:witness}.  While this is technicality for the case
of $d$ constant, this approach could be used to obtain bounds
for super-constant values of $d$.

The basic approach is not new, and has been used in other settings,
such as \cite{ABKU,MitzenmacherThesis}.  Essentially, we can repeat
the ``layered induction'' approach of \cite{ABKU} in the setting of
double hashing, making use of the results of Section~\ref{sec:fluid}
that the deviations from the fully random setting are at most $o(1)$
for a suitable number of levels.

This allows us to state the following theorem:
\begin{theorem} Suppose $n$ balls are placed into $n$ bins using
the balanced allocation scheme with double hashing.
Then with $d \geq 3$ choices (for $d$ constant) the maximum load is $\log \log n / \log d + O(1)$
with high probability.
\end{theorem}

\begin{proof}
Let $z_i$ be the number of bins of load $i$ after all $n$ balls have
been thrown.  We will follow the framework of the original balanced allocations paper \cite{ABKU},
and start by noting that $z_6 \leq n/(2e)$.  Now from the argument of
Section~\ref{sec:fluid}, the probability that the $t$th ball chooses $d$ bins 
all with load at least $i \geq 2$ is bounded above by $z_{i-1}^d/n^d + \eta$,
where $\eta = O(d^2 \log^{2} n)/n$ was determined in Lemma~\ref{lem:mainlemma},
as long as, up to that point, we can condition on all the ancestry
lists being suitably small, which is a high probability event.  We will denote the event 
that the ancestry lists are suitably small throughout the process by ${\cal E}_0$.  

Finally, let $\beta_6 = n/(2e)$ and $\beta_i = 4\beta_{i-1}^d/n^{d-1}$ for $i \geq 6$.  
Let ${\cal E}_i$ be the event that ${\cal E}_0$ occurs and that $z_i \leq \beta_i$.  
(We choose $\beta_i$ values similarly to \cite{ABKU} for convenience, but use the constant
4 on the right hand side whereas \cite{ABKU} uses the constant $e$ to account for
the extra $\eta$ in our probability over just the value $z_{i-1}^d/n^d$.)  
A simple induction using the formula for $\beta_i$ yields $\beta_i \leq n/e^{d^{i-6}}$
for $d \geq 3$.  

Now we fix some $i > 6$ and consider random variables $Y_t$, where $Y_t =
1$ if the following conditions all hold:
all $d$ choices for the $t$th ball have load at least $i-1$, the number of bins with
load at least $i-1$ before the ball is thrown is at most $\beta_{i-1}$, and the ancestry lists
are all suitably small when the ball is thrown so the polylogarithmic bound on the ``extra probability''
that a ball ends up with all $d$ choices having load at least $i-1$ holds.  Let $Y_t =0$ otherwise.  We note
that the number of bins with load at least $i$ is at most the sum of the $Y_t$.
Let $p_i = \beta_{i-1}^d/n^d + \eta$.  
Conditioned on ${\cal E}_{i-1}$, we have 
$$\Pr(z_i \geq k~|~{\cal E}_{i-1}) \leq \Pr(\sum_t Y_t \geq k~|~{\cal E}_{i-1}) \leq \frac{\Pr(\sum_t Y_t \geq k)}{\Pr({\cal E}_{i-1})}.$$
Now the sum $Y_t$ are dominated by a binomial random variable $B(n,p_i)$ of $n$ trials, each with probability $p_i$ of success, because of the definition of the $Y_i$.  

As in \cite{ABKU}, we can use the simple Chernoff bound from \cite{AS}
$$\Pr(B(n,p_i) \geq ep_i n) \leq e^{-p_i n}).$$
Note that, for large enough $n$ and $\beta_{i-1}$, $ep_i n \leq 4\beta_{i-1}^d/n^d$, as $\eta$
will be a lower order term.  
Hence for such values, 
$$\Pr(B(n,p_i) \geq \beta_i) \leq e^{-p_i n}).$$
With these choices, we see that as long as $p_i \geq n^{-1/5}$ (note that for this value of $p_i$,
$\eta$ is indeed a lower order term), 
$$\Pr(\neg {\cal E}_{i}~|~{\cal E}_{i-1}) \leq e^{-n^{4/5}} / \Pr({\cal E}_{i-1}),$$
and using 
$$\Pr(\neg {\cal E}_{i}) \leq \Pr(\neg {\cal E}_{i}~|~{\cal E}_{i-1})\Pr({\cal E}_{i-1}) + \Pr(\neg{\cal E}_{i-1}),$$
we have 
$$\Pr(\neg {\cal E}_{i}) \leq e^{-n^{4/5}} + \Pr(\neg {\cal E}_{i-1}).$$
Recall again that ${\cal E}_{6}$ depended on ${\cal E}_{0}$ and 
$z_6 \leq \beta_6 =n/(2e)$, and the latter holds with certainty.  

Note that we only require $i^* = \log \log n/ \log d +O(1)$ before $p_i \leq n^{-1/5}$,
based on the bound for the $\beta_i$.
Hence the total probability that the required events ${\cal E}_{i}$ do not hold up to this
point is bounded by $\Pr(\neg{\cal E}_{0}) + O(\log \log n)\cdot e^{-n^{4/5}}.$
Hence, as long $\Pr(\neg{\cal E}_{0})$ is $1-o(1)$ (which we argued in Section~\ref{sec:fluid}),
we are good for loads up to $i^*$.  
After only one more round, using the same argument, we can get to the point where 
$z_{i^*+1} \leq n^{5/6}$, using the same Chernoff bound argument, since the expected
number of bins with load at least $i^*+1$ would be dominated by $n^{4/5}$.  

{F}rom this point, one can show that the maximum load is $i^* + c$ for
some constant $c$ with high probability by continuing with a variation of the {\em layered
  induction} argument as used in \cite{ABKU}.  If we condition on
there being $n^{1-\zeta}$ bins with load at least $i'$ for some $i' \geq
i^*$, for a ball have all $d$ choices have bins with at least $i'+1$, it must have at
least two of its bin choices have load at least $i'$.  Even when using
double hashing, for any ball, any pair of the $d$ choices of bins are
chosen independently from all possible pairs of distinct
bins\footnote{Here we again assume $n$ is prime; if not, we need to
  take into account the issue that the offset is relatively prime to
  $n$.}; hence, by a union bound the probability any ball causes a bin
to have load at least $i'+1$ is at most ${d \choose 2}n^{-2\zeta}$,
giving an expected number of bins of load at least $i'+1$ of at most
${d \choose 2}n^{1-2\zeta}$.  (Here this step is slightly different than
the corresponding step in \cite{ABKU}; because of the use of double hashing in
place of independent hashes, we use a union bound over the ${d \choose
  2}$ pairs of bins.  This avoids the issue of the ancestry lists completely at this 
point of the argument,
which we take advantage of once we've gotten down to a small enough number of bins to complete
the argument.)  

Applying the same Chernoff bounds as previously, we find 
$z_{i^*+2} \leq en^{2/3}$ with high probability, 
$z_{i^*+3} \leq e^2n^{1/3}$ with high probability.
By a union bound, the probability of any ball having at least 2 choices
with load at least $i^*+4$ is at most $n \cdot  (e^2 n^{-2/3})^2 = o(1)$, 
and hence $z_{i^*+4} = 0$ with probability $1-o(1)$.  
Note can make the probability smaller (such as $1-o(1/n)$) by taking a larger constant $O(1)$ term.  
This gives that the maximum load is $\log \log n / \log d + O(1)$
with high probability under double hashing.  
\end{proof}
\end{document}